\documentclass{amsart}
\usepackage{subfigure}

\author{Maria Bras-Amor\'os, Klara Stokes}
\thanks{This work was partly supported by
the Spanish Government through projects TIN2009-11689
``RIPUP'' and CONSOLIDER INGENIO 2010 CSD2007-00004 ``ARES'',
 and by the Government of Catalonia under grant 2009 SGR 1135.}

\title{The Semigroup of Combinatorial Configurations}

\usepackage{amssymb}
\usepackage{amsthm}
\usepackage{gastex}
\usepackage{graphics}
\usepackage[usenames]{color}

\newtheorem{lemma}{Lemma}

\newtheorem{corollary}{Corollary}
\newtheorem{theorem}{Theorem}
\newcommand\mut[1]{\ignorespaces}

\begin{document}
\maketitle

\begin{abstract}
We elaborate on the existence and construction of the so-called
combinatorial configurations. The main result is that
for fixed degrees the existence of such configurations 
is given by a numerical semigroup.
The proof is constructive giving a method to obtain combinatorial configurations
with parameters large enough.
\end{abstract}

\section{Introduction}

In the literature
a combinatorial configuration~\cite{Gropp07}  (or a partial linear space
\cite{CTM})
is defined as a 
particular case of a so-called incidence structure. 
Here, for simplicity, we chose to
define it as a particular case of bipartite graph.
Since incidence structures and bipartite graphs are essentially the same,
our choice to use bipartite graphs does not introduce any ambiguity.
We define a $(v,b,r,k)$-combinatorial configuration as a connected bipartite graph 
with $v$ vertices on one side, each of 
them of degree $r$, 
and $b$ vertices on the other side, 
each of them of degree $k$, and with no cycle of length~$4$.

There are many results on the existence of combinatorial configurations.
For instance in Gropp's papers 
\cite{Gropp92,Gropp94,Gropp05,Gropp07} and in Gr\"unbaum's book \cite{Grunbaum}.  
Gropp states in his references
that the next two conditions are necessary for the existence of a $(v,b,r,k)$-configuration.
P1: $vr=bk$; P2: $v\geq r(k-1)+1$.
In particular, for $k=3$ he proves that P1 and P2 are also sufficient. 
The
next theorem 
by Gropp
guarantees the existence of large configurations and, in fact, 
the existence of any configuration satisfying the necessary conditions with {\it sufficiently large} $v$ (and so $b$). 
Its
limitation 
is the restriction on the choice of the parameters $r,k$.

\begin{theorem}\cite[Theorem 3.13]{Gropp94}
For given $k$ and $r$ with $r=tk$
there is a $v_0$ depending on $k,t$ such that there is a 
$(v,b,r,k)$-configuration for all $v\geq v_0$
satisfying P1 and P2.
\end{theorem}

In this short note we will generalize this result by showing 
that 
for any fixed $k\geq 2$ and for {\em any}
fixed $r\geq 2$,
the set of tuples $(v,b)$ for which a $(v,b,r,k)$-combinatorial 
configuration exists
is in bijection with a numerical semigroup. 

\section{The submonoid of $(r,k)$-configurable tuples}

We say that the tuple $(v,b,r,k)$  is {\it configurable}
if a $(v,b,r,k)$ configuration exists.
%
It is immediate to prove 
that if $(v,b,r,k)$ is configurable then $vr=bk$ and consequently
there exists $d$ such that $v=d\frac{k}{\gcd(r,k)}$ and
$b=d\frac{r}{\gcd(r,k)}$.
So, to each configurable tuple $(v,b,r,k)$ we can assign an integer $d$ 
and two different configurable tuples 
$(v,b,r,k)$ will have different integers $d$.
Let us call $D_{r,k}$ the set of all possible integers $d$ corresponding to configurable tuples
$(v,b,r,k)$.
That is,
$$D_{r,k}=\{d\in{\mathbb N}_0:(d{\small \frac{k}{\gcd(r,k)}},d{\small \frac{r}{\gcd(r,k)}},r,k)\mbox{ is configurable}\}.$$

Our aim is to study $D_{r,k}$.
We will consider the empty graph to be also a configuration and
consequently $0\in D_{r,k}$ for all pair $r,k$.
Obviously $D_{r,k}=D_{k,r}$ and $D_{1,k}=\{0,k\}$.
We will prove that if $r,k>1$ then $D_{r,k}$ is a numerical semigroup, that is, 
a subset of ${\mathbb N}_0$ containing $0$,
closed under addition and with a finite complement in ${\mathbb N}_0$.
A general reference on numerical semigroups is \cite{RoGa:llibre}.
If $a_1,\dots,a_l$ are coprime then the set
$\{n_1a_1+\dots+n_la_l:n_1,\dots,n_l\in{\mathbb N}_0\}$
is a numerical semigroup and it is called the semigroup generated by $a_1,\dots,a_n$ 
and denoted by $\langle a_1,\dots,a_l\rangle$. 

In the next section we will give a complete description
of $D_{2,k}$ and in the last one we
will study the case $r\geq 3$.

\section{The case $r=2$}

There is a natural bijection between $(v,b,2,k)$-configurations
and $k$-regular connected graphs with $b$ vertices and $v$ edges.
Two vertices in the graph share an edge if and only if the corresponding nodes
in the configuration share a neighbor and viceversa. The next well-known lemma 
is the key result for describing $D_{2,k}$.

\begin{lemma}
\label{lemma:r2keven}
\begin{enumerate}
\item
If $k$ is even,
a connected $k$-regular graph with $b$ vertices 
exists if and only if $b\geq k+1$.
\item
If $k$ is odd,
a connected $k$-regular graph with $b$ vertices 
exists if and only if $b$ is even and $b\geq k+1$.
\end{enumerate}
\end{lemma}

\begin{proof}
\begin{enumerate}
\item
By definition, any $k$-regular graph must have a number of vertices at least $k+1$.
Conversely, suppose $b\geq k+1$.
Consider a set of vertices $x_1,\dots,x_b$.
Put an edge between $x_i$ and $x_j$, with $i\leq j$, if $j-i\leq k/2$
or $i+b-j\leq k/2$. This gives a 
connected $k$-regular graph with $b$ vertices.
\item
By definition, any $k$-regular graph must have a number of vertices at least $k+1$.
Now, since the number of edges is $kb/2$ this means that $kb$ must be even and since 
$k$ is odd $b$ must be even.
Conversely, suppose $b$ is even and $b\geq k+1$.
Consider a set of vertices $x_1,\dots,x_b$.
Put an edge between $x_i$ and $x_j$, with $i\leq j$, if $j-i\leq (k-1)/2$
or $i+b-j\leq (k-1)/2$. Put also 
edges between $x_i$ and $x_{i+b/2}$ for $i$ from $1$ to $b/2$.
This gives a 
connected $k$-regular graph with $b$ vertices.
\end{enumerate}
\end{proof}

\begin{corollary}
$D_{2,k}=\left\{
\begin{array}{ll}
\langle k+1,k+2,\dots,2k+1\rangle
&\mbox{if }k\mbox{is even,}
\\
\left\langle\frac{k+1}{2},\frac{k+1}{2}+1,\frac{k+1}{2}+2,\dots,k\right\rangle
&\mbox{if }k\mbox{is odd.}
\end{array}
\right.$
\end{corollary}

\section{The case $r\geq 3$, $k\geq 3$}


\subsection{The set $D_{r,k}$ is non-trivial}

\begin{lemma}[Sachs \cite{Sachs}]
\label{lemma:n-reg}
For any integer $n\geq 3$ and any 
$\gamma\geq 2$ there exists an 
$n$-regular graph with girth at least $\gamma$.
\end{lemma}

\begin{lemma}
\label{lemma:Drkmot0}
For any pair of integers $r,k$,
there exists at least one non-zero integer in $D_{r,k}$ for all $r,k$.
\end{lemma}

\begin{proof}
The cases in which $r\leq 2$ or $k\leq 2$
have been proved in the previous sections.
So, we can assume that $r\geq 3$ and $k\geq 3$.
Consider the complete bipartite graph $K_{r,k}$.
From basic graph theory we know that
we can take a subset 
of $r+k-1$ edges in $K_{r,k}$ such that
they connect all $r+k$ vertices and no cycle
is formed (i.e., a generating tree).
Let $A$ be the set of the $rk-r-k+1$
remaining edges
of $K_{r,k}$.

Let $n=rk-r-k+1$ be the number of edges in $A$. 
Notice that since $r$ and $k$ are at least $3$ then $n\geq 3$.
Consider an $n$-regular graph $G$ with girth at least $5$ as in 
Lemma~\ref{lemma:n-reg}
and consider as many copies of
$K_{r,k}$ as vertices in $G$.
Associate each copy of $K_{r,k}$ to a different vertex in $G$.
For each edge $e$ in $G$, 
take the copies of the graphs $K_{r,k}$ 
corresponding to the ends of $e$ and
swap one edge $xy$ in $A$ in the first copy and
one edge $x'y'$ in $A$ in the second copy 
for $xy'$ and $x'y$
(here we abused notation using the same letter $A$ for different copies of
it).
This can be done in a way such that every time we take one edge in $A$ 
corresponding to a given 
copy of $K_{r,k}$, the edge is different.

It is easy to check that we obtain a non-trivial 
$(r,k)$-biregular bipartite graph
with girth at least $5$.
\end{proof}

\subsection{The set $D_{r,k}$ is a numerical semigroup}

\begin{lemma}
\label{lem:separacio} 
Suppose we have a $(v,b,r,k)$-configuration with $r,k\geq 2$.
There exist three edges in the configuration such that
the six ends are all different.
\end{lemma}

\begin{proof}
Since no cycle of length $4$ exists and $r,k\geq 2$,
there exists a path with four edges with the five ends being different.
Three of these ends will be on one partition of the graph while the other two
will be in the other partition.
Take the vertex at the end of the path.
It must be one of the three in the same partition.
Since its degree is at least $2$,
then it will have one neighbor not in the path.
So, by adding the edge from the end of the path to this additional vertex,
we obtain a new path with $5$ edges with all its vertices being different.
By taking the first, third, and fifth edges of this new path we obtain the result.
\end{proof}

This lemma tells us that the vertices $\{x_1,\dots,x_v\}$, $\{y_1,\dots,y_b\}$
 in a $(v,b,r,k)$-configuration with $r\geq 3$ can be arranged in a way such that
the edges $x_1y_1$, $x_2y_2$ and $x_vy_b$ belong to the configuration.

Suppose we have a $(v,b,r,k)$-configuration
with vertices
$\{x_1,\dots,x_v\}$, $\{y_1,\dots,y_b\}$
 and
a $(v',b',r,k)$-configuration
with vertices
$\{x'_1,\dots,x'_v\}$, 
$\{y'_1,\dots,y'_b\}$.
Consider the graph with vertices
$\{x_1,\dots,x_v\}\cup\{x'_1,\dots,x'_v\}$, $\{y_1,\dots,y_b\}\cup\{y'_1,\dots,y'_b\}$
and all the edges in the original configurations.
Swap the edges $x_vy_b$ and $x'_1y'_1$ for $x_vy'_1$ and $x'_1y_b$. 
This gives a $(v+v',b+b',r,k)$ configuration
\cite{DoBrWuMa}.
This construction proves the next lemma.

\begin{lemma}
\label{lemma:construccio}
If $(v,b,r,k)$, $(v',b',r,k)$ are configurable tuples, so is 
$(v+v',b+b',r,k)$.
\end{lemma}

\begin{lemma}
\label{lemma:submonoid}
$D_{r,k}$ satisfies
\begin{itemize}
\item $0\in D_{r,k}$
\item If $d,d'\in D_{r,k}$ then $d+d'\in D_{r,k}$.
\end{itemize}
\end{lemma}

\begin{proof}
Obviously $0\in D_{r,k}$
and, by Lemma~\ref{lemma:construccio}, if
 $d,d'\in D_{r,k}$ then $d+d'\in D_{r,k}$.
\end{proof}

In order to have a numerical semigroup it remains to see that
the number of elements in ${\mathbb N}_0\setminus D_{r,k}$
is finite. This will be proved in the next theorem.
In the proof of the theorem it is used that 
two coprime integers generate a numerical semigroup and so, if a subset containing $0$
and closed under addition contains two coprime integers then it is a numerical semigroup.

\begin{theorem}
\label{theorem:Drksemigroup}
$D_{r,k}$ is a numerical semigroup.
\end{theorem}

\begin{proof}
Because of the results in the previous sections we can assume that $r$ and $k$ are at least $3$.
By Lemma~\ref{lemma:Drkmot0} and since $D_{r,k}\subseteq{\mathbb N}$,
there is a minimal non-zero element $m$ in $D_{r,k}$.
Let us call $v=mk/\gcd(r,k)$ and $b=mr/\gcd(r,k)$.
Select a 
$(v,b,r,k)$ 
configuration.
Take $s=rk/\gcd(r,k)$ 
copies of this configuration.
Let us call the vertices of the $i$th copy
$x^{(i)}_1,\dots,x^{(i)}_v$, $y^{(i)}_1,\dots,y^{(i)}_b$.
By Lemma \ref{lem:separacio} 
we can assume that
$x^{(i)}_1y^{(i)}_1$, $x^{(i)}_2y^{(i)}_2$ and $x^{(i)}_vy^{(i)}_b$ belong to the 
$i$th copy.
Consider $k/\gcd(r,k)$ further vertices $x'_1,\dots,x'_{k/\gcd(r,k)}$
and $r/\gcd(r,k)$ further vertices $y'_1,\dots,y'_{r/\gcd(r,k)}$.
For all $i<s$ swap the edges $x^{(i)}_vy^{(i)}_b$ and $x^{(i+1)}_1y^{(i+1)}_1$ 
for $x^{(i)}_vy^{(i+1)}_1$ and $x^{(i+1)}_1y^{(i)}_b$. 
Remove the edges $x^{(i)}_2y^{(i)}_2$ for all $i\leq s$.
Add the edges
$$x'_1y^{(1)}_2, x'_1y^{(2)}_2, \dots, x'_1y^{(r)}_2,$$
$$x'_2y^{(r+1)}_2, x'_2y^{(r+2)}_2, \dots, x'_2y^{(2r)}_2,$$
$$\vdots$$
$$x'_{k/\gcd(r,k)}y^{(s-r+1)}_2,\dots,x'_{k/\gcd(r,k)}y^{(s)}_2$$
and
$$x^{(1)}_2y'_1,x^{(2)}_2y'_1,\dots,x^{(k)}_2y'_1,$$
$$x^{(k+1)}_2y'_2,x^{(k+2)}_2y'_2,\dots,x^{(2k)}_2y'_2,$$
$$\vdots$$
$$x^{(s-k+1)}_2y'_{r/\gcd(r,k)},\dots,x^{(s)}_2y'_{r/\gcd(r,k)}.$$

It is easy to check that this is a new configuration
with parameters
\begin{eqnarray*}\left(sv+\frac{k}{\gcd(r,k)},sb+\frac{r}{\gcd(r,k)},r,k\right)
&=&\left(\frac{(sm+1)k}{\gcd(r,k)},\frac{(sm+1)r}{\gcd(r,k)},r,k\right)
\end{eqnarray*}
and so $sm+1\in D_{r,k}$.

Since $m$ and $sm+1$ are coprime, they generate a numerical semigroup
and this semigroup 
is contained in $D_{r,k}$. So the complement of $D_{r,k}$ in ${\mathbb N}_0$
is finite and $D_{r,k}$ is a numerical semigroup.
\end{proof}

As a consequence of the fact that the necessary conditions P1, P2 are also sufficient
for $k=3$ it is easy to deduce that 
$D_{r,k}=\{0\}\cup(\frac{2r+1}{3}\gcd(3,r)+{\mathbb N}_0)$.
The computation of examples for $r,k>3$ is 
computationally very hard.

%

\section{Conclusion}
The main conclusion of this short note is that
for fixed $r$ and $k$ there exist configurations
for all parameters $b,v$ large enough provided that $vr=bk$.

Another important fact is that our proofs are all constructive and so 
we can derive algorithms for constructing large configurations.

\bibliographystyle{plain}

\end{document}